\documentclass[twoside,a4paper,12pt]{article}
\pdfoutput=1

\usepackage[english]{babel}
\usepackage{amsmath}
\usepackage{amssymb}
\usepackage{amsthm}
\usepackage{mathrsfs}
\usepackage{graphicx}
\newcommand{\comment}[1]{}

\newcommand{\n}{\noindent}
\newcommand{\be}{\begin{equation}}
\newcommand{\ee}{\end{equation}}
\renewcommand{\Im}{\operatorname{Im}\,}
\renewcommand{\Re}{\operatorname{Re}\,}

\newcommand{\diam}{\operatorname{diam}}

\newcommand{\slim}{\operatorname{s}-\lim}

\newcommand{\RE}{\mathbb R}
\newcommand{\CO}{\mathbb C}
\newcommand{\II}{\mathbb I}
\newcommand{\HH}{\mathcal H}
\newcommand{\OO}{\mathcal O}
\newcommand{\ve}{\varepsilon}
\newcommand{\de}{\delta}
\newcommand{\ga}{\gamma}
\newcommand{\bt}{\beta}
\newcommand{\al}{\alpha}
\newcommand{\la}{\lambda}
\newcommand{\om}{\omega}
\newcommand{\si}{\sigma}
\newcommand{\Ga}{\Gamma}

\newtheorem{theorem}{Theorem}
\newtheorem{proposition}{Proposition}

\newtheorem{definition}{Definition}
\theoremstyle{remark}
\newtheorem{remark}{Remark}
\newcommand{\XX}{\chi}

\newcommand{\hh}{\hat H}
\newcommand{\hR}{\hat R}

\newcommand{\BB}{\mathcal B}

\title{Resonances in  Models of Spin Dependent Point Interactions}
\author{C. Cacciapuoti$^1$, R. Carlone$^2$, R. Figari$^3$}

\date{}

\begin{document}

\maketitle
\begin{center}
\vspace{0,5cm}
$^1$Czech Technical University, Doppler Institute.\\
B\v{r}ehov\'a 7, 11519 Prague, Czech Republic.\\
\vspace{0,5cm}
$^2$Doppler Institute,\\
Nuclear Physics Institute, Czech Academy of Sciences.\\
25068 \v{R}e\v{s} near Prague, Czech Republic.\\
\vspace{0,5cm}
$^3$Istituto Nazionale di Fisica Nucleare (INFN), Sezione di Napoli.\\
Dipartimento di Scienze Fisiche, Universit\`a di Napoli Federico II.\\
Via Cintia 80126 Napoli, Italy.\\
\end{center}
\n E-mail: cacciapuoti@ujf.cas.cz, carlone@ujf.cas.cz, figari@na.infn.it.

\vspace{1cm}

\begin{abstract}
In dimension $d=1,2,3$ we define a family of two-channel Hamiltonians obtained as point perturbations of the generator of the free decoupled dynamics. Within the family we choose two Hamiltonians,  $\hat H_0$ and $\hat H_\ve$, giving rise respectively to the unperturbed and to the perturbed evolution.  The Hamiltonian $\hat H_0$ does not couple the channels and has an eigenvalue embedded in the continuous spectrum. The Hamiltonian  $\hat H_\ve$ is a small perturbation, in resolvent sense, of $\hat H_0$ and exhibits a  small coupling between the channels.

We take advantage of the complete solvability of our model to prove with simple arguments that the embedded eigenvalue of $\hat H_0$ shifts into a resonance for $\hat H_\ve$. In dimension three  we analyze details of the time behavior of the  projection onto the region of the spectrum close to the resonance.   
\end{abstract}


\section{Introduction\label{introduction}}
Aim of the paper is to analyze a completely solvable quantum system where formation and time decay properties of a metastable state can be explicitly investigated. The main technical tool we shall make use of is the theory of point interaction Hamiltonians in a slightly generalized setting. The Hamiltonians we define in the paper appeared frequently, in different forms, in the chemical and physical literature. In \cite{DO88} they are extensively utilized, without a formal classification, to compute scattering data for low energy inelastic scattering of electrons by atoms (see also \cite{LY02} for more recent applications of those ideas). In the theory of neutron scattering by nuclei zero range potentials with spin dependent strength were often exploited (see, e.g., \cite{Lov84}). In the framework of the modern theory of point interaction Hamiltonians (see  \cite{AGH-KH05} for a thorough  introduction)  particles with spin coupled with internal degrees of freedom and/or with  magnetic fields were investigated by various authors (\cite{BGP07},  \cite{ES07} and references therein).

Much closer to ours are the works of C. King (\cite{Kin91}, \cite{Kin94}) (whose relation with what is done here will be clarified in the last section), and of P. Exner \cite{Exn91}. At the best of our knowledge the latter author was the only one to characterize multichannel point interaction Hamiltonians in a rigorous way.  In  \cite{Exn91} he analyzes thoroughly  scattering theory for a two channel point interaction in three dimensions. Using different techniques, recently made available to the theory of self-adjoint extensions of symmetric operators, in this paper we give an unified presentation in any dimension $d=1,2,3$ of the model-Hamiltonians introduced by Exner  and we add few results on the decay properties of metastable states.  

In a recent paper \cite{CCF07b} we characterized a family of self-adjoint Hamiltonians generating the dynamics of a quantum particle interacting, via zero range forces, with an array of localized quantum systems with a finite number of energy levels.
Preliminary results suggest that in those systems it is possible to investigate details of the reduced dynamics of a particle moving through a quantum environment. In particular in \cite{CCF07a} we used such kind of Hamiltonians to examine the evolution of a quantum particle inside a tracking chamber.

On the other hand the same family of Hamiltonians may be used to describe quantum systems where the role of the environment is played by the particles. Following this idea we define and analyze a system made up of a localized quantum bit (a
model-atom or a spin) in interaction with a non relativistic quantum particle. We first define an unperturbed system showing a
ground state level and an upper energy level embedded in the continuous part of the spectrum. We then show how a small
perturbation makes the ground state move slightly in energy and the upper level turn into a resonance. One can verify that any characterization of resonances apply to our model. The analysis of resonances and their perturbations for Schr\"odinger operators with short range potentials and for point interaction Hamiltonians was given in \cite{AH-K84b}, where one can find also a list of references about the various aspects of resonance theory exhaustive up to the time of its publication. More recently a renewed interest on the subject, both for Schr\"odinger and wave equation, showed up (see, e.g., \cite{Bal89},  \cite{Hun90},  \cite{MS99}, \cite{SW98}). Our work deals with resonances generated by perturbation of eigenvalues embedded in the continuum. For recent results on this specific aspect of resonance theory see \cite{CS01}. For theory and applications about resonances induced by curvature in quantum waveguides see \cite{DEM01}, \cite{DES95}, \cite{KS07} and references therein.

For the sake of simplicity we consider here only the case of a two level atom and we will investigate only metastables states obtained as a perturbation of embedded eigenvalues not at the continuous threshold. Results on the latter case will be given elsewhere. The extension of our results to a generic discrete array of any multilevel localized system is straightforward. With a suitable choice of  positions, energy levels and  interaction parameters it is possible to
define Hamiltonians with a very rich spectral structure. In particular the final Hamiltonian can exhibit any number of
metastable and isolated bound states as well as any number of thresholds. On the other hand any matrix-valued Schr\"odinger operator with smooth potential could be approximated by point interaction Hamiltonians of the type we define in this paper with minimal changes with respect to what was done in the scalar case \cite{FHT88}.

The generalization to the case of a large number of non relativistic bosons interacting with the localized q-bit is in
progress, but no modifications in the main qualitative features of the results presented here are expected. It may have some pedagogical interest the analogy which exists  between this latter system with the one consisting of a non relativistic atom coupled to the
radiation field, analyzed in a series of important papers in the last decade (see \cite{A-SFFS07pp},  \cite{BFP06}, \cite{HHH07pp} and references therein). The correspondence with this much more complex quantum system can be established interpreting the ionization of a particle in the model we analyze in this paper as the
creation of a photon and the localized quantum system as the atom. ``Spontaneous emission'' would correspond to the transition from
a metastable state to a state where the atom gets a lower energy and an ionized particle is produced. A local attractive
interaction between the atom and the particles produces a ``vacuum state'' well separated from the continuous part of the spectrum
preventing any ``infrared problem'' in  the model-system considered here. The case of a ground state
at the continuum threshold will be treated elsewhere. 

The analogy suggests that, in spite of their simplicity, the Hamiltonians we
analyze show interesting spectral characteristics considered to be typical of more elaborate and realistic systems.

\section{Basic Notation and Preliminary Results}

This section will be devoted to introduce notation and to recall few basic results about Hamiltonians we will make use of in the rest of the paper.

The system under analysis is made up of one (non-relativistic) quantum particle and one two level localized subsystem  (q-bit or
spin).

The two level system will be described as a spin 1/2  placed in a fixed position of space, i.e. a unitary vector in
$\CO^2$. Without loss of generality we assume that it is placed in the origin. We denote by  $\hat\sigma^{(1)}$ the first Pauli
matrix and with  $\chi_{\sigma}\in\CO^2$ the normalized eigenvectors of the operator $\hat\sigma^{(1)}$ \begin{equation*}
\hat\sigma^{(1)}\chi_{\sigma}=\sigma\chi_{\sigma}\qquad \sigma=\pm1\,;\;\|\chi_{\sigma}\|_{\CO^2}=1\,. \end{equation*} A spin state can be
written as a linear superposition $a\,\chi_++b\,\chi_-$, with $a,b\in\CO$ and $|a|^2+|b|^2=1$.

The Hilbert state-space of a system made up of one particle in dimension $d$ and one spin 1/2 is \begin{equation*}
\HH:=L^2(\RE^d)\otimes\CO^2\qquad d=1,2,3\,. \end{equation*}

A capital Greek letter will denote a vector in $\HH$. Any $\Psi\in\HH$ admits  a decomposition of  the following form
\begin{equation*}
%\label{decomposition}
\Psi=\sum_{\sigma}\psi_\sigma\otimes\chi_\sigma\,, 
\end{equation*}
where the sum runs over $\si=\pm$. In this
representation the scalar product in $\HH$ reads
\begin{equation*}
\langle\Psi,\Phi\rangle:=\sum_{\sigma}(\psi_{\sigma},\phi_{\sigma})_{L^2}
\qquad\Psi,\Phi\in\HH\,.
\end{equation*}

Let $S$ be the linear operator in  $\HH$ whose domain and action are defined as follows
\begin{equation*}
D(S):=C_0^\infty(\RE^d\backslash\underline0)\otimes\CO^2\qquad d=1,2,3
\end{equation*}
\begin{equation*}
S:=-\Delta\otimes\mathbb{I}_{\CO^2}+\mathbb{I}_{L^2}\otimes\beta\hat\sigma^{(1)}\qquad\beta\in\RE^+\,,
\end{equation*}
where
$C_0^\infty(\RE^d\backslash\underline0)$ denotes the space of infinitely differentiable functions in $\RE^d$ with compact
support not containing the origin. $S$ is a densely defined symmetric operator in $\HH$. A trivial self-adjoint extension of $S$
is the operator
\begin{equation*}
%\label{DH}
D(H):=H^2(\RE^d)\otimes\CO^2\qquad d=1,2,3
\end{equation*}
\begin{equation*}
%\label{H}
H:=-\Delta\otimes\mathbb{I}_{\CO^2}+\mathbb{I}_{L^2}\otimes\beta\hat\sigma^{(1)}\qquad\beta\in\RE^+\,,
\end{equation*}
where $H^2(\RE^d)$
denotes  the standard Sobolev space of functions in $L^2(\RE^d)$ together with their first and second generalized derivatives.
$H$ describes the free evolution of a quantum particle and a localized spin whose reduced dynamics is generated by the
Hamiltonian $\bt\hat\si^{(1)}$. Notice that no interaction between particle and spin is considered at this stage. We fixed
$\hbar=2m=1$, where $\hbar$ is the Plank's constant and $m$ denotes the particle mass. The constant $\bt$ has the dimension of an
energy which, in our units, corresponds to a reciprocal square length and it measures half the energy gap between the spin
states $\chi_+$ and $\chi_-$. The action of $H$ on vectors in its domain can be written as \begin{equation*}
H\Psi=\sum_{\sigma}\big(-\Delta+\beta\,\sigma\big) \psi_\sigma\otimes\chi_\sigma \qquad\Psi\in D(H)\,. \end{equation*} Its resolvent
$R(z):=(H-z)^{-1}$  acts  on $\HH$ as 
\begin{equation*}
R(z)\Psi=\sum_{\sigma}\big(-\Delta-z+\beta\,\sigma\big)^{-1}
\psi_\sigma\otimes\chi_\sigma\qquad\Psi\in\HH;\; z\in\rho(H)\,, 
\end{equation*} 
\n where $\rho(H)$ denotes  the resolvent set of $H$. The
integral kernel of  $\big(-\Delta-w\big)^{-1}$, $G^w(x-x')$, is given by
\begin{equation}
\label{Gz}
G^w(x):=\left\{\begin{aligned}
&i\frac{e^{i\sqrt{w}|x|}}{2\sqrt{w}}
& d=1\\
&\frac{i}{4}H^{(1)}_0\big(\sqrt{w}\,|x|\big)\qquad& d=2
\\
&\frac{e^{i\sqrt{w}|x|}}{4\pi|x|}& d=3
\end{aligned}\right.
\qquad\textrm{with}\quad w\in\CO\backslash\RE^+\,;\;\Im(\sqrt{w})>0
\end{equation}
Here $H^{(1)}_0\big(\eta\big)$ is the zero-th Bessel
function of third kind (also known as Hankel function of the first kind). We recall that $H^{(1)}_0\big(\eta\big)$ tends to zero
as $|\eta|\to \infty$ for $\Im\eta>0$ and that it has a logarithmic singularity in zero
\begin{equation*}
H^{(1)}_0\big(\eta\big)=
\frac{2i}{\pi}\ln\frac{\eta}{2}+1+\frac{2i\ga}{\pi}+\OO(\ln(\eta)\eta^2)\,,
\end{equation*}
where $\ga$ is the Euler's constant. The spectrum of
$H$ is  only absolutely continuous\footnote{In this paper we will adopt the classification of the spectrum given in
\cite{RSI}.}, in particular
\begin{equation*}
\sigma_{pp}(H)=\varnothing\;;\quad \sigma_{ess}(H)=\sigma_{ac}(H)=[-\beta,\infty)\,.
\end{equation*}

All the self-adjoint extensions of $S$ can be obtained by Krein's theory. Details of the construction in dimension one and three
can be found in \cite{CCF07b}. For the sake of completeness we give in the following a brief account of the procedure.

Let $S^*$ denote the adjoint of $S$, the deficiency indices of $S$ are defined as $n_\pm=\dim[\ker(S^*\pm i)]$. It is easily
seen that: $(n_+,n_-)=(4,4)$ for $d=1$ and $(n_+,n_-)=(2,2)$ for $d=2,3$. Since the whole family of self-adjoint extensions of
$S$ is parameterized by the family of unitary maps between $\ker(S^*+i)$ and $\ker(S^*-i)$, in dimension one  there is a 16 real
parameters family of self-adjoint extensions of $S$,  while in dimensions two and three the number of free parameters is reduced
to 4. For $d \geqslant  4$ the deficiency indices are $(n_+,n_-)=(0,0)$ meaning that the operator $S$ is essentially
self-adjoint and its closure coincides with the free Hamiltonian $H$.

Let us define 
\begin{equation*}
%\label{phiz2d}
\Phi_{\sigma}^z:=G^{z-\beta\,\sigma}\otimes\XX_{\sigma}\qquad z\in\CO\backslash\RE\,,
\end{equation*}
where $G^w(x)$, has been defined in (\ref{Gz}). For $d=1,2,3$,  $\Phi_{+}^z$ and $\Phi_{-}^z$ are independent solutions of the
equation
\begin{equation}
\label{defecteq}
S^*\Phi_\si^z=z\Phi_\si^z\,\qquad \si=\pm\,;\;\Phi_\si\in D(S^*)\,;\;z\in\CO\backslash\RE\,.
\end{equation}

In dimension two and three $\{\Phi_{+}^z,\Phi_{-}^z\}$ is a basis of $\ker(S^*-z)$. In dimension one there exist two more
independent solutions of equation \eqref{defecteq}, namely $\Phi_{1\sigma}^{z}:=(G^{z-\beta\,\sigma})'\otimes\XX_{\sigma}$,
$\si=\pm$, where $(G^{z-\beta\,\sigma})'$ is the derivative of $G^{z-\beta\,\sigma}$. In such a case a
basis of the defect space is $\{\Phi_{+}^z,\Phi_{-}^z,\Phi_{1+}^z,\Phi_{1-}^z\}$.

We will make use  of $2\times2$ matrices; their elements  will be labeled according to the following notation
\begin{equation*}
 M:=
\begin{pmatrix}
 M_{(+,+)}&M_{(+,-)}\\ M_{(-,+)}&M_{(-,-)}
 \end{pmatrix}\,.
 \end{equation*}
  Given two matrices $M_1$ and $M_2$ we will denote
by $(M_1|M_2)$ the $2\times 4$ matrix with the  first and second column given by the columns of $M_1$ and third and forth column
given by the columns of $M_2$.

The following theorem holds true
\begin{theorem}
\label{thse}
Let $A$ and $B$ be $2\times2$ matrices such that $AB^*=BA^*$ and
$(A|B)$ is of maximal rank . Define the operator
\begin{align*}
D(H^{AB}):=\Big\{& \Psi=\sum_\si\psi_\si\otimes\XX_\si\in\HH\Big|\; \Psi =\Psi^z+\sum_{\si}q_{\si}\Phi_{\si}^z\,;\;\Psi^z\in
D(H);\;z\in\rho(H^{AB})\,;\nonumber\\
&\sum_{\si'} A_{\si,\si'}q_{\si'} =\sum_{\si'}B_{\si,\si'}f_{\si'}\,;
\\
&\begin{aligned}
&f_\si=\psi_\si(0)&&d=1\,;\\
&f_\si=\lim_{|x|\rightarrow 0}\Big[\psi_\si(x)+\frac{q_\si}{2\pi}\ln(|x|)\Big]\quad&&d=2\,;\\
&f_{\si}=\lim_{|x|\rightarrow 0}\Big[\psi_\si(x)-\frac{q_{\si}}{4\pi\,|x|}\Big]&&d=3\,\Big\}
\end{aligned}
\end{align*}
\begin{equation}
\label{action3d}
H^{AB}\Psi:=H\Psi^z+z\sum_{\si}q_{\si}\Phi_{\si}^z \qquad\Psi\in D(H^{AB})\,.
\end{equation}
$H^{AB}$ is
self-adjoint and its resolvent, $R^{AB}(z)=(H^{AB}-z)^{-1}$, is given by
\begin{equation}
\label{resolvent}
R^{AB}(z)=R(z)+\sum_{\si,\si',\si''}
\big((\Gamma^{AB}(z))^{-1}\big)_{\si,\si'}B_{\si',\si''} 
\langle\Phi^{\bar z}_{\si''} ,\cdot\,\rangle\Phi^z_{\si}\quad
z\in\rho(H^{AB})\,. 
\end{equation} 
\n Where $\Gamma^{AB}(z)$ is the $2\times 2$ matrix \begin{equation*} \Gamma^{AB}(z)=B\Gamma(z)+A\,, \end{equation*} \n with
\begin{align*}
&\Gamma(z)=
\begin{pmatrix}
\displaystyle-\frac{i}{2\sqrt{z-\bt}}&0\\
0&\displaystyle-\frac{i}{2\sqrt{z+\bt}}
\end{pmatrix}\qquad&& d=1\\
&\Gamma(z)=
\begin{pmatrix}
\displaystyle\frac{\ln(\sqrt{z-\bt}/2)+\ga-i\pi/2}{2\pi}&0\\
0&\displaystyle\frac{\ln(\sqrt{z+\bt}/2)+\ga-i\pi/2}{2\pi}
\end{pmatrix}&& d=2\\
&\Gamma(z)=
\begin{pmatrix}
\displaystyle\frac{\sqrt{z-\bt}}{4\pi i}&0\\
0&\displaystyle\frac{\sqrt{z+\bt}}{4\pi i}
\end{pmatrix}&& d=3\,.
\end{align*}
\end{theorem}
\begin{remark}
For $d=2,3$, theorem \ref{thse} gives all possible self-adjoint extensions of $S$. For $d=1$, Hamiltonians $H^{AB}$ include only
extensions of $S$ for which the wave function part of vectors in $D(H^{AB})$ is continuous in the origin. To cover the whole
family of self-adjoint extensions of $S$, and include Hamiltonians defined on vectors with discontinuous wave function, we
should have taken into account the whole bases of the defect spaces of $S$. For the analysis of the general case one can
refer to \cite{CCF07b}.
\end{remark}
\begin{remark}
\label{rem}
It is easily seen that for  $\Im z>0$ the operator $R^{AB}(z)-R(z)$ is a finite rank operator. This implies that $\sigma_{ess}(H^{AB})=\sigma_{ess}(H)=[-\bt,\infty)$ (see, e.g., \cite{BS86}). Moreover 
notice that the pure point spectrum of $H$ is empty. In turn  this means that any $\lambda\in\sigma_{pp}(H^{AB})$ must be a polar  singularity of $R^{AB}(z)-R(z)$.
\end{remark}
\begin{proof} The
proof of theorem  \ref{thse} for  $d=1,3$ was given, in a more general setting, in \cite{CCF07b}. We give here a sketch of the
proof only  for $d=2$. Define two linear applications
$\Lambda:D(S^*)\to\CO^2$ and $\tilde\Lambda:D(S^*)\to\CO^2$, $\Lambda$ defines the ``charges'' $q_\si$, 
\begin{equation*}
(\Lambda\Psi)_\si:=q_\si=-\lim_{|x|\to0}\frac{2\pi}{\ln(|x|)}\psi_\si(x)\qquad \Psi=\sum_\si\psi_\si\otimes\chi_\si\in D(S^*)\,.
\end{equation*} 
$\tilde\Lambda$  defines $f_{\si}$, i.e.,  the value in the origin of the regular part of the wave function,
\begin{equation*}
(\tilde\Lambda\Psi)_{\si}:=f_{\si} \qquad\Psi=\sum_{\si}\psi_\si\otimes\chi_{\si}\in D(S^*)\,. 
\end{equation*}
A direct calculation shows
that  the triple $(\CO^m,\Lambda,\tilde\Lambda)$  is a boundary value space for $S$ (see, e.g., \cite{GG91}), i.e. for all
$\Psi_1,\,\Psi_2\in D(S^*)$ \begin{equation*} \langle\Psi_1,S^*\Psi_2\rangle-\langle S^*\Psi_1,\Psi_2\rangle
=\sum_{\si}\Big[\overline{(\Lambda\Psi_1)}_\si(\tilde\Lambda\Psi_2)_\si-
\overline{(\tilde\Lambda\Psi_1)}_\si(\Lambda\Psi_2)_\si\Big] \end{equation*} moreover  $\Lambda$ and $\tilde\Lambda$ are surjective. Then
all self-adjoint extensions of $S$ are given by the  restrictions of $S^*$ on vectors $\Psi$ satisfying \begin{equation*}
\sum_{\si'}A_{\si,\si'}(\Lambda\Psi)_{\si'}= \sum_{\si'}B_{\si,\si'}(\tilde\Lambda\Psi)_{\si'}\,, \end{equation*} \n where $A$ and $B$ are
two $2\times 2$ matrices satisfying $AB^*=BA^*$ and $(A|B)$ has maximal rank (see, e.g., Theorem 3.1.6 in \cite{GG91}). This
proves that operators $H^{AB}$ are self-adjoint. The resolvent formula \eqref{resolvent} comes directly from  the resolvent
formula in \cite{AP05} (see also Theorem 10 in \cite{Pan06}).  Formula \eqref{action3d} is a consequence of the fact that
$H^{AB}\Psi=\big(R^{AB}(z)\big)^{-1}\Psi+z\Psi$.
\end{proof}

\section{The Unperturbed Hamiltonian}

Our analysis will proceed as follows. Within the family of point interaction Hamiltonians described in the previous section we
pick up an ``unperturbed'' Hamiltonian $\hh_0$ which exhibits one eigenvalue embedded in the continuous spectrum.

In the next section we then show that inside the same family one can find a sub-family of self-adjoint modifications  of $\hh_0$
indexed by a perturbation parameter $\ve$. We prove that for $\ve>0$ the embedded eigenvalue moves off the real axis and we
estimate the time decay rate of the corresponding metastable state. For a general approach to the problem of resonances
generated by perturbation of eigenstates embedded in the continuous spectrum see \cite{CS01}.

\begin{definition} Let $-\infty<\al<\infty$. The operator $\hh_0$ is  defined as follows 
\begin{equation*}
\begin{aligned}
D(\hh_0):=\Big\{\Psi\in\HH\Big|&\; \Psi
=\Psi^z+\sum_{\sigma}q_{\sigma}\Phi_{\sigma}^z\,;\;\Psi^z=\sum_\sigma\psi_\sigma^z\otimes\chi_\sigma\in
D(H);\;z\in\rho(\hh_0)\,;\\
&q_{\sigma}=-\al f_\si\,,\;d=1\,;\;\al q_\si=f_\si\,,\;d=2,3\Big\}
\end{aligned}
\end{equation*} \begin{equation*} \hh_0\Psi:=H\Psi^z+z\sum_{\sigma}q_{\sigma}\Phi_{\sigma}^z \qquad\Psi\in D(\hh_0)\,. \end{equation*}
\end{definition}
Self-adjointness of $\hh_0$ comes directly from theorem \ref{thse} by taking $A=\II_2$ and $B=-\al\II_2$  for $d=1$,
$A=\al\II_2$ and $B=\II_2$ for $d=2,3$.

Theorem \ref{thse} gives also the explicit formula  for resolvent of $\hh_0$
\begin{equation*}
\hR_0(z)=R(z)+\sum_{\sigma,\sigma'}
\big((\Gamma_0(z))^{-1}\big)_{\sigma,\sigma'} \langle\Phi^{\bar z}_{\sigma'} ,\cdot\,\rangle\Phi^z_{\sigma}\qquad
z\in\CO\backslash\RE\,, \end{equation*} where
\begin{align*}
&\Gamma_0(z)=
\begin{pmatrix}
\displaystyle-\frac{i}{2\sqrt{z-\bt}}-\frac{1}{\al}&0\\
0&\displaystyle-\frac{i}{2\sqrt{z+\bt}}-\frac{1}{\al}
\end{pmatrix}\qquad&& d=1\\
&\Gamma_0(z)=
\begin{pmatrix}
\displaystyle\frac{\ln\big(\sqrt{\mu(z-\bt)}\big)-i\pi/2}{2\pi}&0\\
0&\displaystyle\frac{\ln\big(\sqrt{\mu(z+\bt)}\big)-i\pi/2}{2\pi}
\end{pmatrix}&& d=2\\
&\Gamma_0(z)=
\begin{pmatrix}
\displaystyle\frac{\sqrt{z-\bt}}{4\pi i}+\al&0\\
0&\displaystyle\frac{\sqrt{z+\bt}}{4\pi i}+\al
\end{pmatrix}&& d=3\,,
\end{align*}
where $\mu=e^{2(\gamma+2\pi\al)}/4$.

\n Notice that the matrix elements of $\Gamma_0$ have the dimension of the inverse of the square root of an energy in $d=1$, the dimension of the square root of an energy in $d=3$ whereas they are dimensionless in $d=2$.

The spectrum of $\hh_0$ is characterized as follows
\begin{theorem}
For $d=1,2,3$ the  essential  spectrum  of  $\hh_0$ is \begin{equation*} \sigma_{ess}(\hat H_0)=[-\beta,+\infty)\,.
\end{equation*}

The point spectrum coincides with the set of the (real) roots of equation $\det\Gamma_0(z)=0$.\\
\begin{itemize}
    \item[$d=1$.] If $0 \leqslant\al< \infty$ the point spectrum is empty. If $-\infty<\al<0$ the point spectrum consists of two simple eigenvalues given by
\begin{equation*}
 E_{0,-}=-\bt-\frac{\al^2}{4}\;;\qquad E_{0,+}=\bt-\frac{\al^2}{4}\,.
\end{equation*} For all $-\infty<\al<0$ the lowest eigenvalue, $ E_{0,-}$, is below the threshold of the essential spectrum. For
$\al<-2\sqrt{2\bt}$ also the second eigenvalue, $ E_{0,+}$, is below the threshold of the essential spectrum, in this case
the point spectrum is only discrete and the essential spectrum is only absolutely continuous. If $-2\sqrt{2\bt}\leqslant\al<0$
the upper eigenvalue is embedded in the continuous spectrum, $-\bt\leqslant E_{0,+}<\bt$.
\item[$d=2$.] For all 
$0<\mu<\infty$
the point spectrum consists of two simple eigenvalues given by
\begin{equation*}
 E_{0,-}=-\bt-\frac{1}{\mu}\;;\qquad E_{0,+}=\bt-\frac{1}{\mu}\,.
\end{equation*}
The lowest eigenvalue, $ E_{0,-}$, is always below the threshold of the essential spectrum. For
$0<2\beta\mu<1$
 also the second eigenvalue, $ E_{0,+}$, is below the threshold of the essential
spectrum, in this case the point spectrum is only discrete and the essential spectrum is only absolutely continuous. If $1\leqslant2\beta\mu<\infty$ the upper eigenvalue is embedded in the continuous spectrum, $-\bt\leqslant
E_{0,+}<\bt$.
\item[$d=3$.] If $0 \leqslant\al <\infty$ the point spectrum is empty. If $-\infty<\al<0$ the point spectrum consists of two simple eigenvalues given by
\begin{equation*}
 E_{0,-}=-\bt-(4\pi\al)^2\;;\qquad E_{0,+}=\bt-(4\pi\al)^2\,.
\end{equation*} The lowest eigenvalue, $ E_{0,-}$, is always below the threshold of the essential spectrum. For
$-\infty<\al<-\sqrt{2\bt}/(4\pi)$ also the second eigenvalue, $ E_{0,+}$, is below the threshold of essential spectrum, in this case the point spectrum is only discrete and the essential spectrum is only absolutely continuous. If
$-\sqrt{2\bt}/(4\pi)\leqslant\al<0$ the upper eigenvalue is embedded in the continuous spectrum, $-\bt\leqslant E_{0,+}<\bt$.
\end{itemize}
\end{theorem}
\begin{proof}
The proof of theorem 2 is straightforward. As it was  stressed in remark \ref{rem}  the essential spectrum of $\hh_0$ coincides with the essential spectrum of $H$. The pure point spectrum of  $\hh_0$  can only be a subset of the set of the  real singularities of $\hR_0(z)-R(z)$, which coincides with the  real solutions of $\det\Gamma_0(z)=0$.

In fact, being $\Gamma_0(z)$ diagonal, it is immediate to check that the total spectrum of $\hh_0$ is the union of two components obtained shifting, respectively of $\beta$ or $-\beta$, the spectrum of a point interaction Hamiltonian $h_{\alpha}$ of ``strength'' $\alpha$ (see \cite{AGH-KH05}). In fact we have \begin{equation*}
\hh_0=(h_\al+\bt)\otimes\Pi_+ + (h_\al-\bt)\otimes\Pi_-\,,
\end{equation*} 
where $\Pi_\pm$ are the projectors on $\chi_+$ and $\chi_-$ respectively, $\Pi_\pm=(\chi_\pm,\cdot)_{\CO^2}\chi_\pm$. The dynamics generated by $\hh_0$ describes a particle evolving in a point potential in two independent channels characterized by a different energy of the spin.
\end{proof}
Let us denote by $\hat Q_0:L^2(\RE^d)\to \RE$ the quadratic form associated to the operator $\hat H_0$. $\hat Q_0$ is defined as follows: 
\begin{itemize}
\item[$d=1$.] 
\begin{equation*}
D(\hat Q_0):=\Big\{\Psi\in\HH\Big|\; \Psi=\sum_\sigma\psi_\sigma\otimes\chi_\sigma\,;\;\psi_\sigma\in H^1(\RE)\Big\}
\end{equation*}
\begin{equation*}
\hat Q_0[\Psi]:=\sum_{\sigma}\bigg[\int_\RE|\psi_\sigma'|^2+\sigma\beta|\psi_\sigma|^2dx+\alpha|\psi_\sigma(0)|^2\bigg]\,.
\end{equation*}
\item[$d=2,3$.] 
\begin{equation*}
\begin{aligned}
D(\hat Q_0):=\Big\{\Psi\in\HH\Big|&\; \Psi=\sum_\sigma\psi_\sigma\otimes\chi_\sigma\,;\;
\psi_\sigma=\psi_\sigma^\lambda+q_\sigma G^{\lambda-\sigma\beta}\,;\\
&\psi_\sigma^\la\in H^1(\RE^d)\,,\;q_\si\in\CO\,,\;\la<-\beta\Big\}
\end{aligned}
\end{equation*}
\begin{equation*}
\hat Q_0[\Psi]:=\sum_{\sigma}\bigg[\int_{\RE^d}|\nabla\psi^\la_\sigma|^2-(\la-\si\beta)|\psi^\la_\sigma|^2+\la|\psi_\sigma|^2dx+\big(\Gamma_0(\la)\big)_{\si,\si}|q_\si|^2\bigg]\,.
\end{equation*}
\end{itemize}

\section{The Perturbed Hamiltonian}
The perturbed dynamics is generated by the Hamiltonian $\hh_\ve$ defined as follows
\begin{definition} Let $-\infty<\al<\infty$ and $0<\ve\ll|\al|$. The Hamiltonian  $\hh_\ve$ has the following domain and action
\begin{equation*}
\begin{aligned}
D(\hh_\ve):=\Big\{\Psi\in\HH\Big|&\; \Psi
=\Psi^z+\sum_{\sigma}q_{\sigma}\Phi_{\sigma}^z\,;\;\Psi^z=\sum_\sigma\psi_\sigma^z\otimes\chi_\sigma\in
D(H);\;z\in\rho(\hh_\ve)\,;\\
&\begin{aligned}
&q_{\pm}=-\al f_\pm-\ve f_\mp&&d=1\,;\\
&\al q_\pm+\ve q_\mp=f_\pm\qquad&&d=2,3\;\Big\}
\end{aligned}
\end{aligned}
\end{equation*} \begin{equation*} \hh_\ve\Psi:=H\Psi^z+z\sum_{\sigma}q_{\sigma}\Phi_{\sigma}^z \qquad\Psi\in D(\hh_\ve)\,. \end{equation*}
\end{definition}
Again self-adjointness is a direct consequence of theorem \ref{thse}. Notice that $\hh_\ve$ has been obtained by adding $-\ve$
(resp. $\ve$) to the off diagonal terms of  matrix $B$ (resp. $A$) used in the definition of $\hh_0$ for $d=1$ (resp. $d=2,3$).

The resolvent of $\hh_\ve$ is \begin{equation*} \hR_\ve(z)=R(z)+\sum_{\sigma,\sigma'} \big((\Gamma_\ve(z))^{-1}\big)_{\sigma,\sigma'}
\langle\Phi^{\bar z}_{\sigma'} ,\cdot\,\rangle\Phi^z_{\sigma}\qquad z\in\CO\backslash\RE\,, \end{equation*} where
\begin{align*}
&\Gamma_\ve(z)=
\begin{pmatrix}
\displaystyle-\frac{i}{2\sqrt{z-\bt}}-\frac{\al}{\al^2-\ve^2}&\displaystyle\frac{\ve}{\al^2-\ve^2}\\ \\
\displaystyle\frac{\ve}{\al^2-\ve^2}&\displaystyle-\frac{i}{2\sqrt{z+\bt}}-\frac{\al}{\al^2-\ve^2}
\end{pmatrix}\qquad &d=1\\
&\Gamma_\ve(z)=
\begin{pmatrix}
\displaystyle\frac{\ln\big(\sqrt{\mu(z-\bt)}\big)-i\pi/2}{2\pi}&\ve\\ \\
\ve& \displaystyle\frac{\ln\big(\sqrt{\mu(z+\bt)}\big)-i\pi/2}{2\pi}
\end{pmatrix}
&d=2\\
&\Gamma_\ve(z)=\begin{pmatrix}
                \displaystyle-\frac{i\sqrt{z-\bt}}{4\pi}+\al&\ve\\ \\
                \ve&\displaystyle-\frac{i\sqrt{z+\bt}}{4\pi}+\al
                \end{pmatrix}
&d=3
\end{align*}

It is easy to verify that for all $z\in\CO\backslash\RE$, there exists $\ve_0$ such that for all $0<\ve<\ve_0$
\begin{equation}
\label{resext}
\|\hR_\ve(z)-\hR_0(z)\|_{\BB(\HH,\HH)}\leqslant \ve C
\end{equation}
where $C$ is a positive constant\footnote{In the following we shall always denote by $C$ a generic positive constant. The value of $C$ can change from line to line.}. Formula \eqref{resext} implies the uniform convergence of the resolvent of $\hh_\ve$ to the resolvent of $\hh_0$ for $\ve\to0$. Consequently (see, e.g., \cite{RSI}) the unitary group $e^{-i\hh_\ve t}$ converges strongly to $e^{-i\hh_0 t}$ uniformly for $t$ in any finite interval. In this sense Hamiltonian $\hh_\ve$ is a small perturbation of $\hh_0$.

We also point out the operators $\hat H_0$ and $\hat H_\ve$ have the same form domain. In particular, let us denote by $\hat Q_\ve$ the quadratic form  associated to the operator $\hat H_\ve$, then $D(\hat Q_\ve)=D(\hat Q_0)$ and 
\begin{itemize}
\item[$d=1$.] 
\begin{equation*}
\hat Q_\ve[\Psi]=\hat Q_0[\Psi]+\ve\big[ \bar\psi_+(0)\psi_-(0)+ \bar\psi_-(0) \psi_+(0)\big]\,,
\end{equation*}
\item[$d=2,3$.] 
\begin{equation*}
\hat Q_\ve[\Psi]=\hat Q_0[\Psi]+\ve\big[ \bar q_+q_-+ \bar q_-q_+\big]\,.
\end{equation*}
\end{itemize}
It is quite easily seen that the difference  $\hat Q_\ve-\hat Q_0$ is form-bounded with respect to $\hat Q_0$.

In the following theorem we prove that for $\ve$ small enough Hamiltonian $\hh_\ve$ does not have
embedded eigenvalues even when $\hh_0$ does. In particular we show that  the real pole of the resolvent of $\hh_0$ between $-\bt$ and $\bt$ moves in a pole of $\hR_\ve$ in the unphysical region of the complex energy. In fact, due to the presence of the square roots $\sqrt{z-\bt}$ and $\sqrt{z+\bt}$, the resolvents $\hR_0(z)$ and $\hR_\ve(z)$ can be extended to operator analytic functions on the complex plane with two branch-cuts starting respectively from $z=\bt$ and $z=-\bt$. The Riemann surfaces then exhibit four sheets corresponding to the four couples of branches chosen for the square roots. We will prove that the pole belongs to the sheet defined by  $\Im(\sqrt{z-\bt})>0$ and $\Im(\sqrt{z+\bt})<0$, that is, the second Riemann sheet of $\sqrt{z+\bt}$. The proof of the theorem is based on a direct analysis of the singularities of the resolvent.
\begin{theorem}
\label{teorema3}
Assume that 
\begin{equation}
\begin{aligned}
&-2\sqrt{2\bt}<\al<0\qquad& d=1\\
&1<2\beta\mu<\infty\qquad &d=2\\
&-\sqrt{2\bt}/(4\pi)<\al<0\qquad& d=3
\label{alcond3d}
\end{aligned}
\end{equation}
 then  there exists $\ve_0>0$ such that for all $0<\ve<\ve_0$ the  essential  spectrum of $\hat H_\ve$ is only absolutely continuous,
\begin{equation*}
\sigma_{ess}(\hat H_\ve)=\sigma_{ac}(\hat H_\ve)=[-\beta,+\infty)\,.
\end{equation*}
Moreover for all $0<\ve<\ve_0$ the point spectrum consists of a single eigenvalue, $E_{\ve,-}<-\bt$, and 
the analytic continuation of the resolvent $\hR_\ve(z)$ through the real axis from the semi-plane $\Im z >0$  has  a simple pole (resonance) in $z= E_{\ve}^{res}$ where $\Im (E_{\ve}^{res})<0$.
For some positive constant $C$ one has $|E_{\ve}^{res}-E_{0,+}|\leqslant C\ve^2$
\end{theorem}
\begin{proof}
From the compactness of $\hat R_\ve(z)-R(z)$ for any $z\in\CO\backslash\RE$  it follows that the essential spectrum of $\hh_\ve$ coincides with $[-\bt,\infty)$. According to theorem 6.10 in \cite{HS96} if  $\la$ is an embedded eigenvalue for $\hh_\ve$ the projection onto the subspace corresponding to the eigenvalue $\la$ is given by 
\[
 P_\lambda=
 \slim_{\delta\to0}(-i\delta)\big(\hh_\ve-\la-i\delta\big)^{-1}\,.
 \]
We shall prove that for all $\ve>0$ and  for all $\la\in[-\bt,\infty)$ the projection $P_\la$ equals zero which, in turn, implies that $\hh_\ve$ has no embedded eigenvalues for all $\ve >0$.
We start proving that $P_\la$ is zero for any $\la \in [-\bt,\infty)$ whenever the matrix $\Gamma_\ve(\la)$ has a bounded inverse. The proof easily follows from the fact that there are no eigenvalues embedded in the continuos spectrum of Schr\"odinger operators with point interactions in any dimension. For the sake of completeness we give here a sketch of the proof (with some detail for $d=3$).

Let 
$\Psi$ belong to $\HH$, $\de>0$ and $\la\in[-\bt,\infty)$. The following inequality holds
\be
\label{mrwalker}
\|\big(\hh_\ve-\la-i\delta\big)^{-1}\Psi\|
\leqslant
\|R(\la+i\delta)\Psi\|+ \sum_{\si,\si'}\big|\big((\Gamma_\ve(\la+i\delta))^{-1}\big)_{\sigma,\sigma'}\big|
|\langle\Phi^{\la-i\delta}_{\sigma'},\Psi\rangle|
\|\Phi^{\la+i\delta}_{\sigma} \|\,.
\ee

The Hamiltonian $H$ has no embedded eigenvalues, which implies that
\[
\lim_{\delta\to0}\de\|R(\la+i\delta)\Psi\|=0\,.
\]

To study the second term on the r.h.s. of equation \eqref{mrwalker} we consider separately the cases $\la\neq\pm\bt$ and $\la=\pm\bt$. Let us assume first $\la\neq\pm\bt$. For all $\zeta\in(0,\infty)$  and for $\de$ small enough the following inequality holds true
\[
\|G^{\zeta+i\de}\|_{L^2(\RE^d)}\
\leqslant \frac{C}{\sqrt{\de}}\qquad{d=1,2,3}\,,
\]
where $C$ is a positive constant depending on $d$  and $\la$ and the functions $G^w$ were defined in \eqref{Gz}. The inequality can be proved by noticing that the Fourier transform of $G^w$ equals $(|k|^2-w)^{-1}$ in any dimension. For all $\la\neq\pm\bt$ and for all $\si$ we have
\be
\label{uno}
0<\lim_{\delta\to0}\de^{1/2}
\|\Phi^{\la+i\delta}_{\sigma} \|<\infty \qquad d=1,2,3\,.
\ee

Moreover it is easy to show that for any $f\in L^2(\RE^d)$, $d=1,2,3$, and $\zeta\in(0,\infty)$
\be
\label{freddo}
\lim_{\de\to0}\de^{1/2}|(G^{\zeta-i\de},f)_{L^2(\RE^d)}|=0
\ee
which implies that for all $\Psi\in\HH$, $\sigma$,  and $\la\neq\pm\bt$
\be
\label{due}
\lim_{\de\to0}\de^{1/2}|\langle\Phi^{\la-i\delta}_{\sigma},\Psi\rangle|=0\,.
\ee
In fact take $d=3$,  $R>0$ and let $B(R/\de^b)$ be the open ball  of radius $R/\de^b$ centered in the origin, with $0<b<1$. We have  
\[
\begin{aligned}
\delta^{\frac{1}{2}}
|(G^{\zeta-i\de},f)_{L^2(\RE^3)}|\leqslant
&\delta^{\frac{1}{2}}\int_{B(R/\de^b)}\frac{|f(x)|}{4\pi|x|}dx+\delta^{\frac{1}{2}}\int_{\RE^3\backslash
B(R/\de^b)}
\frac{e^{-\Im(\sqrt{\zeta-i\delta})|x|}}{4\pi|x|}|f(x)|dx\\
\leqslant& \Big(\frac{R}{4\pi}\Big)^{\frac{1}{2}}\delta^{\frac{(1-b)}{2}}\|f\|_{L^2(\RE^3)}+C\bigg(\int_{\RE^3\backslash
B(R/\de^b)}|f(x)|^{2}dx\bigg)^{\frac{1}{2}}\,,
\end{aligned}
\]
from which formula \eqref{freddo} follows.

An analogous result holds true in $d=1,2$.

From equations \eqref{uno} and \eqref{due} it follows that for all $\la\neq\pm\bt$
\be
\label{tre}
\lim_{\de\to0}\de\big|\big((\Gamma_\ve(\la+i\delta))^{-1}\big)_{\sigma,\sigma'}\big|
|\langle\Phi^{\la-i\delta}_{\sigma},\Psi\rangle|
\|\Phi^{\la+i\delta}_{\sigma'} \|=0
\ee
whenever the matrix $(\Gamma_\ve(z))^{-1}$ has no singularities in $z=\la$. 

Let us analyze  the cases $\la=\pm\bt$. The following inequality holds
\[
\|G^{i\de}\|_{L^2(\RE^d)}\
\leqslant C\frac{1}{\de^{1-d/4}}\qquad{d=1,2,3}
\]
 where $C$ is a positive constant which depends on $d$.  We then have
 \be
\label{quattro}
 \|\Phi^{\mp\bt+i\de}_\pm\|\leqslant C\frac{1}{\de^{1-d/4}}\qquad{d=1,2,3}\,.
\ee
 
 Moreover for all $f\in L^2(\RE^d)$, 
 \[
 |(G^{-i\de},f)|_{L^2(\RE^d)}\
\leqslant C\frac{1}{\de^{1-d/4}}\|f\|_{L^2(\RE^d)}\qquad{d=2,3} 
 \]
 while, proceeding in the same way as it was done in the proof of formula \eqref{freddo} one can prove that
 \[
 \lim_{\de\to0}
\de^{3/4}|(G^{-i\de},f)|_{L^2(\RE)}=0\qquad d=1\,.
 \]
We then have 
 \be
 \label{5}
|\langle\Phi^{\mp\bt-i\de}_\pm,\Psi\rangle|\
\leqslant C\frac{1}{\de^{1-d/4}}\|\Psi\|\qquad{d=2,3} 
 \ee
and
\be
 \label{6}
\lim_{\de\to0}
\de^{3/4}
|\langle\Phi^{\mp\bt-i\de}_\pm,\Psi\rangle|=0\qquad d=1
 \ee
for all $\Psi\in\HH$. 

Since for $d=3$, and $\ve$ small enough,  the matrix $(\Gamma_\ve(z))^{-1}$ is regular in $z=\pm\bt$,  inequalities \eqref{quattro} and \eqref{5} imply that equation \eqref{tre} holds for $\la=\pm\bt$ 
in $d=3$.
 
In $d=1,2$ the statement follows from formulae \eqref{quattro}, \eqref{5} and \eqref{6},  noticing that 
 \[
 \lim_{\de\to0}
|\det\Ga_\ve(\pm\bt+i\de)|^{-1}=0\qquad d=2
 \]
 and 
 \[
0< \lim_{\de\to0}
\frac{|\det\Ga_\ve(\pm\bt+i\de)|^{-1}}{\de^{1/2}}<\infty\qquad d=1\,.
 \]
 We are finally left to prove that $(\Gamma_\ve(z))^{-1}$ is regular for any $\la\in[-\bt, \infty)$ which in turn means that the equation $\det\Ga_\ve(z)=0$ has no solutions in $[-\bt,\infty)$.

\n The equation $\det\Ga_\ve(z)=0$ for $d = 1,2,3$ reads
\begin{equation*}
%\label{eqris1}
\big(i(\al^2-\ve^2)+2\al\sqrt{z-\bt}\big)\big(i(\al^2-\ve^2)+2\al\sqrt{z+\bt}\big)-4\ve^2\sqrt{z-\bt}\sqrt{z+\bt}=0\qquad d=1
\end{equation*}
\begin{equation*}
%\label{eqris2}
\big(\ln\big(\sqrt{\mu(z-\bt)}\big)-i\pi/2\big)\big(\ln\big(\sqrt{\mu(z+\bt)}\big)-i\pi/2\big)-(2\pi\ve)^2=0 \qquad d=2
\end{equation*}
\begin{equation}
\label{eqris3}
\big(-i\sqrt{z-\bt}+4\pi\alpha\big)\big(-i\sqrt{z+\bt}+4\pi\alpha\big)-(4\pi\ve)^2=0 \qquad d=3
\end{equation}
In the following the spectral properties of $\hat{H}_\ve$ 
are analyzed for $\ve$ ``small enough'',  which will mean $\ve/|\alpha|\ll 1$. We give details only for the case $d=3$. 

Let us define the function 
\begin{equation*}
f(z):=
\big(-i\sqrt{z-\bt}+4\pi\alpha\big)\big(-i\sqrt{z+\bt}+4\pi\alpha\big)\,.
\end{equation*}
Equation  $\det\Ga_\ve(z)=0$ is equivalent to $f(z)=(4\pi\ve)^2$. We shall write $z=\la$ whenever $z$ is real. If $\al>0$, $\Im f(\la)<0$ for all $\la\geqslant -\bt$. For $\al<0$ it is easily seen that for $\la\geqslant\bt$, $\Im f(\la)>0$. It remains to consider the case  $\al<0$ and $-\bt<\la<\bt$. Being  
\begin{equation*}
f'(E_{0,+})=\frac{1}{8\pi|\al|}(i\sqrt{2\beta -(4\pi\al)^2}+4\pi|\al|)\,,
\end{equation*}
$f'(E_{0,+})(z-E_{0,+})$ takes real positive values only if $\Im z<0$ and $\Re z>E_{0,+}$. This means that the solution of $f(z)=(4\pi\ve)^2$, if such a solution exists, moves in the second Riemann sheet for $z+\beta$ in the corresponding square root. An identical  argument for $d=1$ proves the same result. For $d=2$ the imaginary part always moves toward negative values  while the real part of the resonance is greater (smaller) than $E_{0,+}$ for $E_{0,+}<0$ ($E_{0,+}>0$).

To prove that in the interval $-\sqrt{2\bt}/(4\pi)\leqslant\al<0$ there is only one eigenvalue below $-\bt$ we only need to analyze equation \eqref{eqris3} for $\la<-\bt$, for such values of $\la$ the equation  is real. We pose $\al=-|\al|$, then the solutions of equation \eqref{eqris3} are given by 
\begin{equation*}
\big(\sqrt{\bt-\la}/(4\pi|\alpha|)-1\big)\big(\sqrt{-\bt-\la}/(4\pi|\alpha|)-1\big)=\frac{\ve^2}{|\al|^2}\qquad\la<-\bt\,.
\end{equation*}
If condition \eqref{alcond3d} is satisfied,  the function on the left hand side of the equation is positive if and only if $\la<-\bt-(4\pi\al)^2$. Moreover in that interval it is a strictly decreasing function of $\lambda$ approaching $+\infty$ for $\lambda$ tending to $-\infty$ . Then there is one and only one eigenvalue, $E_{\ve,-}$ located   below $-\bt$. Up to the first order in $\ve^2$ the explicit formula for $E_{\ve,-}$ reads
\begin{equation*}
E_{\ve,-}=E_{0,-}-\frac{2(4\pi)^2}{\sqrt{2\bt/(4\pi\al)^2+1}-1}\ve^2+\OO(\ve^4)\,.
\end{equation*}

We are left to prove the last statement of the theorem, i.e., that for $-\sqrt{2\bt}/(4\pi)\leqslant\al<0$,  
$\big(\Gamma_\ve(z)\big)^{-1}$ has an analytic continuation onto the second sheet of $\sqrt{z+\bt}$ across the real axis where it has a pole.

Let us rewrite equation \eqref{eqris3} as
\begin{equation*} {{ \left(\frac{\sqrt{\beta-z}}{4\pi}-|\alpha|\right)}}{{
\left(\frac{\sqrt{\beta+z}}{4\pi\,i}-|\alpha|\right)}}=\ve^2 \,.
\end{equation*}
With the following position
\begin{equation*}
 \xi=\left(\frac{\sqrt{\beta-z}}{4\pi}-|\alpha|\right)\quad;\qquad
 \eta=-\left(\frac{\sqrt{\beta+z}}{4\pi\,i}-|\alpha|\right)\,,
\end{equation*}
the equation reads
\begin{equation*}
\xi=-\frac{\ve^2}{\eta}\,.
\end{equation*}
We look for fixed points of the following recurrence procedure
\begin{equation*}
\eta_0=i\frac{\sqrt{2\beta-(4\pi\alpha)^2}}{4\pi}+|\alpha|\;;\qquad
\xi_k=-\frac{\ve^2}{\eta_k}\;;
\qquad
\eta_k=|\alpha|+i\frac{\sqrt{2\beta-[4\pi(\xi_{k-1}+|\alpha|)]^2}}{4\pi}\;;
\end{equation*}
inside the ball $|\eta-\eta_0|<C\ve^2$. In the definition of $\eta_k$ the square root is the analytic continuation of the square root of a positive number: it is defined with positive real part, while the imaginary part can be positive or negative according to the fact that the argument in the square root is  in the first or second Riemann sheet respectively. It follows that $|\eta_k|>|\alpha|$ and $|\xi_k|<\frac{\ve^2}{|\alpha|}$ and being
\begin{equation*}
\xi_{k+1}-\xi_k=\frac{\ve^2}{\eta_k\,\eta_{k+1}}(\eta_{k+1}-\eta_k)
\end{equation*}
we also have 
\begin{equation}
\label{cinque}
|\xi_{k+1}-\xi_k|\leqslant\frac{\ve^2}{\alpha^2}|\eta_{k+1}-\eta_k|\,.
\end{equation}

We use the following estimate that is a direct consequence of the definition of $\eta_k$
\begin{equation}
\label{sei}
|\eta_{k+1}-\eta_k|
=\left|\frac{\sqrt{2\beta-[4\pi(\xi_{k}+|\alpha|)]^2}}{4\pi}-\frac{\sqrt{2\beta-[4\pi(\xi_{k-1}+|\alpha|)]^2}}{4\pi}\right|\leqslant
C|\xi_{k}-\xi_{k-1}|
\,,
\end{equation}
\n fore some positive constant $C$. From \eqref{cinque} and \eqref{sei}
\begin{equation*}
|\xi_{k+1}-\xi_k|\leqslant\frac{\ve^2}{\alpha^2}|\eta_{k+1}-\eta_k|
\leqslant\frac{\ve^2 C}{\alpha^2}|\xi_{k}-\xi_{k-1}|\,.
\end{equation*}
Summing the series we get
\begin{equation*}
|\xi_{\infty}-\xi_0|\leqslant\frac{|\xi_1-\xi_0|}{1-\frac{\ve^2C}{\alpha^2}}
\end{equation*}
with 
\begin{equation*}
\xi_1-\xi_0=\frac{i |\alpha|}{\sqrt{2\beta-(4\pi\alpha)^2}\left(|\alpha|+\frac{i}{4\pi}\sqrt{2\beta-(4\pi\alpha)^2}\right)^3}\ve^4+\OO(\ve^6)\,.
\end{equation*}
Since
\begin{equation*}
\xi_0=-\frac{\ve^2}{\eta_0}=-\frac{(4\pi)^2\ve^2}{2\beta}\left(|\alpha|-\frac{i\sqrt{2\beta-(4\pi\alpha)^2}}{4\pi}\right)
\end{equation*}
the position of the resonance is
\begin{align*}
E_{\ve}^{res}=&\beta-\left[4\pi(\xi_\infty+|\alpha|)\right]^2\\
=&E_{0,+}+\frac{(4\pi)^4\,\al^2\ve^2}{\beta}-i \frac{(4\pi)^4\,\alpha^2\ve^2}{\beta}\sqrt{2\beta/(4\pi\alpha)^2-1}+\OO(\ve^4)\,.
\end{align*}
Then $|E_\ve^{res}-E_{0,+}|<C\ve^2$ and the negative imaginary part of $E_{\ve}^{res}$ means that the pole is in the second Riemann sheet.

The proof of the theorem for the cases $d=1,2$ does not differ substantially from the one given in the three dimensional case. To the lowest order the eigenvalue below $-\beta$ is given by
\begin{equation*}
\begin{aligned}
&E_{\ve,-}=E_{0,-}-\frac{\ve^2}{2\big(\sqrt{8\beta/\al^2+1}-1\big)}+\OO(\ve^4)\qquad&d=1\\
&E_{\ve,-}=E_{0,-}-\frac{8\pi^2\ve^2/\mu}{\ln\big(\sqrt{2\beta\mu+1}\big)}+\OO(\ve^4)\qquad&d=2
\end{aligned}
\end{equation*}
while the resonance is
\begin{align*}
E_{\ve}^{res}=&E_{0,+}+\frac{\ve^2\al^2}{16\beta}-i\frac{\ve^2}{2}\frac{\sqrt{8\beta/\al^2-1}}{8\beta/\al^2}+\OO(\ve^4)\qquad&d=1\\
E_{\ve}^{res}=&E_{0,+}-\frac{8\pi^2\ve^2/\mu}{a^2+(\pi/2)^2}\bigg(a+i\frac{\pi}{2}\bigg)+\OO(\ve^4)\qquad&d=2
\end{align*}
with $a=\ln\big(\sqrt{2\bt\mu-1}\big)$.
\end{proof}
Notice that for all $d=1,2,3$ the shift of the lowest eigenvalue $E_{0,-}$ is real and negative. The shift of the real part of the resonance for $d=1,3$ is always positive, while for $d=2$ it is positive for $1<2\beta\mu<2$ and negative for $2\beta\mu>2$.\\

In the following we will investigate the characteristic time decay rate associated with the presence of sharp resonances. We will refrain
to give any specific definition of a metastable state and we will limit ourselves to consider the decay rate of the survival probability of a spin-up state for an initial state whose support in energy is contained in a small spectral neighborhood around the position of the eigenvalue of the free Hamiltonian embedded in the continuum.  Generic initial states inside and outside the domain of $H_{\ve}$ would have of course different decaying behaviors around $t=0$.

As it is well known (see, e.g., \cite{EXN}, \cite{Exn07}, \cite{FNP01}, \cite{PN94} and \cite{PNBR93}) the small time decay rate has a great relevance to examine the so called Zeno and anti-Zeno behavior of the initial state. In the following we will not give details regarding any specific initial state. We want only to stress that the knowledge of the generalized eigenfunctions makes available a specific formula for any initial state.

Let us recall that for $\ve$ small enough the essential spectrum of $\hat H_\ve$ is only absolutely continuous and coincides with $[-\bt,\infty)$. Then for $d=1,2,3$ and $\la\geqslant-\bt$ the spectral projection of $\hat H_\ve$ on the absolutely continuous part of the spectrum  can be defined via the Stone's formula as
\begin{equation*}
 P_\ve(\la):=
s-\lim_{\delta\to0^+}\frac{1}{2\pi
  i}\Big[ \hat R_\ve(\lambda-i\delta)-\hat R_\ve(\lambda+i\delta) \Big]\qquad\la\in[-\bt,\infty)\,.
\end{equation*}
Given a state $\Psi$ 
one has that
\begin{equation*}
\Psi(t)=\int_{-\bt}^\infty e^{-i\la t}P_\ve(\la)\Psi\,d\la\,.
\end{equation*}
 An explicit formula for  $P_\ve(\la)$ in terms of the generalized eigenfunctions can be derived in all dimensions. Here we give details only for $d=3$. From a  straightforward calculation it follows that 
\begin{equation*}
 P_\ve(\la)=\sum_{\sigma=\pm}\int_\Omega 
\hat\Phi_\ve^{\sigma}(\la;\omega)
\langle\hat\Phi_\ve^{\sigma}(\la;\omega),\cdot\rangle d\omega
\end{equation*}
where $\Omega$ is the solid angle and
\begin{equation*}
\hat\Phi_{\ve}^\sigma(\la;\omega)=\Phi^{\sigma}(\la;\omega)+
\sum_{\sigma'}\phi^\sigma(\lambda;\omega,0)\,
\big(\Gamma_{\ve,+}(\lambda)\big)^{-1}_{\sigma'\sigma}\,
G_+^{\lambda-\sigma'\alpha} \otimes\chi_{\sigma'}\qquad \la\in[\si\bt,\infty)\,.
\end{equation*}
In the previous formula $\Phi^{\sigma}(\la;\omega)$ are the generalized eigenfunctions of $H$,
\begin{equation*}
\begin{aligned}
\Phi^{\sigma}(\la;\omega):=\phi^\sigma(\la;\om)\otimes\chi_\si\,;\qquad&\phi^\sigma(\la;\om,x)=
\frac{(\lambda-\si\beta)^{\frac{1}{4}}}{4\pi^{\frac{3}{2}}}e^{i\omega\sqrt{\lambda-\si\beta}\,x}\\
&\la\in[\si\bt,\infty)\quad\omega\in\Omega\,.
\end{aligned}
\end{equation*}
while $G_+^{\lambda-\sigma\alpha}$ and $\Gamma_{\ve,+}(\lambda)$ are defined by 
\begin{equation*}
G_+^{\lambda-\sigma\alpha}=
\lim_{\delta\to0^+}G^{\lambda-\sigma\alpha+i\delta}\qquad
\Gamma_{\ve,+}(\lambda)=
\lim_{\delta\to0^+}\Gamma_{\ve}(\lambda+i\delta)\,.
\end{equation*}
The explicit expression of the generalized eigenfunctions of $\hat H_\ve$ reads
\begin{equation*}
\begin{aligned}
\hat\Phi_\ve^{+}(\la;\omega)&=
\frac{(\lambda-\beta)^{\frac{1}{4}}}{4\pi^{\frac{3}{2}}}
\left[e^{i\omega\sqrt{\lambda-\beta}\,\cdot}\otimes\chi_++ \frac{
\frac{\sqrt{\lambda+\beta}}{4\pi i}-|\alpha|}{\left(  \frac{\sqrt{\lambda-\beta}}{4\pi i}-|\alpha|\right)\left(
\frac{\sqrt{\lambda+\beta}}{4\pi i}-|\alpha|\right)-\ve^2}
\frac{e^{i\sqrt{\lambda-\beta}|\cdot|}}{4\pi|\cdot|}\otimes\chi_++\right.\\
&\left.-\frac{\ve}{\left(  \frac{\sqrt{\lambda-\beta}}{4\pi i}-|\alpha|\right)\left(  \frac{\sqrt{\lambda+\beta}}{4\pi
i}-|\alpha|\right)-\ve^2}\frac{e^{i\sqrt{\lambda+\beta}|\cdot|}}{4\pi|\cdot|} \otimes\chi_-\right]\qquad\lambda\geqslant\beta,\,\omega\in\Omega
\end{aligned}
\end{equation*}

\begin{equation*}
\begin{aligned}
\hat\Phi_\ve^{-}(\la;\omega)&=
\frac{(\lambda+\beta)^{\frac{1}{4}}}{4\pi^{\frac{3}{2}}}\left[e^{i\omega\sqrt{\lambda+\beta}\,\cdot}\otimes\chi_-+ \frac{
\frac{\sqrt{\lambda-\beta}}{4\pi i}-|\alpha|}{\left(\frac{\sqrt{\lambda-\beta}}{4\pi i}-|\alpha|\right)\left(
\frac{\sqrt{\lambda+\beta}}{4\pi i}-|\alpha|\right)-\ve^2}
\frac{e^{i\sqrt{\lambda+\beta}|\cdot|}}{4\pi|\cdot|}\otimes\chi_-+\right.\\
&\left.
-\frac{\ve}{\left(\frac{\sqrt{\lambda-\beta}}{4\pi i}-|\alpha|\right)\left(  \frac{\sqrt{\lambda+\beta}}{4\pi
i}-|\alpha|\right)-\ve^2}\frac{e^{i\sqrt{\lambda-\beta}|\cdot|}}{4\pi|\cdot|} \otimes\chi_+\right]
\qquad\lambda\geqslant-\beta,\,\omega\in\Omega
\end{aligned}
\end{equation*}
where we used the fact that $-\sqrt{2\bt}/(4\pi)\leqslant\al<0$ and we posed $\al=-|\al|$. 

We use the following notation for the real and imaginary part of the resonance of  $\hat H_\ve$, $E_\ve^{res}:=b_\ve-i\ga_\ve$. From theorem \ref{teorema3} we have that for $\ve$ small enough $-\bt<b_\ve<\bt$ and $|b_\ve-(\bt-(4\pi \al)^2)|<C\ve^2$, moreover $0<\ga_\ve<C\ve^2$, for some positive constant $C$. 

As we are interested on the decay-rate of a  $\chi_+$-state of the spin we will examine the evolution of the  $\chi_+$ component of the state at any time. Let $I$ be the interval centered in $E_{0,+}$ of amplitude $2 \Delta$, $\ve^2\ll2 \Delta \ll(4\pi \alpha)^2$ . We shall analyze the evolution of the following projector

\begin{equation}
\label{basta}
\mathscr P_{I,++}(t;x,x'):=\frac{1}{64\pi^4}\frac{1}{|x||x'|}
\int_I \frac{\ve^2 \sqrt{\la+\bt}\,e^{-i\la t}\,e^{-\sqrt{\bt-\la}(|x|+|x'|)}}
{\big|\big(\sqrt{\la-\bt}/(4\pi i)-|\al|\big)\big(\sqrt{\la+\bt}/(4\pi i)-|\al|\big)-\ve^2\big|^2}
d\lambda\,.
\end{equation}

Since $\la\in I$,  $\bt-\la>0$  the integrand function in \eqref{basta} is bounded. In turn this implies that the operator $\mathscr P_{I,++}(t)$ is Hilbert-Schmidt for all $t>0$.

Using the substitution $\xi\equiv\frac{\sqrt{\beta-\la}}{4\pi}-|\alpha|$ we can rewrite  (\ref{basta}) as
\[
\mathscr P_{I,++}(t;x,x')=\frac{-\ve^2}{2\pi^2}\frac{e^{-(4\pi)^{2}|\al|(|x|+|x'|)}}{|x||x'|}\,
\int_{I^{'}}\frac{ f(\xi)\,e^{-i\,\left(\beta-(4\pi)^{2}(\xi+|\al|)^{2}\right)\,t}
\,e^{-(4\pi)^{2}\xi(|x|+|x'|)}}
{\xi^{4}+2|\al|\xi^{3}-\frac{2\beta}{(4\pi)^{2}}\xi^{2}-2\ve^{2}|\al|\xi-\ve^{4}}
d\xi
\]
where
\begin{equation*}
f(\xi)= \sqrt{2\bt-(4\pi)^{2}(\xi+|\al|)^{2}}\,(\xi+|\al|)
\end{equation*}
and $I'=\left[\frac{\sqrt{(4\pi|\al|)^{2}-\Delta}}{4\pi}-|\al|,\frac{\sqrt{(4\pi|\al|)^{2}+\Delta}}{4\pi}-|\al|\right]$. The four roots of the denominator of the integrand function are easily analyzed. Two of them are real $\xi_{1}(\ve),\xi_{2}(\ve)\in\mathbb{R}$, close to the roots of  the equation $\xi^{2}+2|\al|\xi-\frac{2\beta}{(4\pi)^{2}}=0$ and are outside the integration region. The last two $\xi_{3}(\ve),\xi_{3}^{*}(\ve)\in\mathbb{C}$ are complex conjugate, close to the roots of the equation $\frac{2\beta}{(4\pi)^{2}}\xi^{2}+2\ve^{2}|\al|\xi+\ve^{4}=0$; the one corresponding to $E_\ve^{res}$ (let us say $\xi_{3}(\ve)$) has positive imaginary part and for $\ve$ small enough $|\xi_{3}(\ve)|\leqslant C\,\ve^2$. With the notation introduced above we can write
\be
\label{subst1}
\mathscr P_{I,++}(t;x,x')=\frac{-\ve^2}{2\pi^2}\frac{e^{-(4\pi)^{2}|\al|(|x|+|x'|)}}{|x||x'|}\,
\int_{I^{'}}\frac{ f(\xi)\,e^{-i\,\left(\beta-(4\pi)^{2}(\xi+|\al|)^{2}\right)\,t}
\,e^{-(4\pi)^{2}\xi(|x|+|x'|)}}
{(\xi-\xi_{1}(\ve))(\xi-\xi_{2}(\ve))(\xi-\xi_{3}(\ve))(\xi-\xi_{3}^{*}(\ve))}
d\xi
\ee

Notice that  the density of states is approximately Lorentzian in $I^{'}$ (for the relevance of a Lorentzian density of states
close to a resonance see \cite{Kin91} and \cite{Kin94}). The difference with a pure Lorentzian behavior (Breit and Wigner for Physicist) is indicated by the presence of a quartic term in the denominator and of a slowly varying function in the numerator of the integrand function. 

In the next theorem we recover the standard result about the time behavior of metastable states. We extract the exponential term and estimate the remainder associated to the non-Lorentzian part of the density of states. 

In the following we shall consider the open subset of $\CO$ defined by $Q=\big\{z\in \CO\big|\; |z|<\diam(I^{'})\,,\;0\leqslant\arg(z)\leqslant\pi\big\}$. The set $Q$ is a semicircle in the upper complex plane, with center in the origin and  having $I^{'}$ as diameter.  

\begin{theorem}
\label{teorema4}
Let  $\mathscr P_{I,++}(t)$ be defined like  in (\ref{basta}), then
\[
\mathscr P_{I,++}(t):=L(t,\ve)+B(t,\ve)
\]
where the following notation was used
\[
L(x,x';t,\ve)=-\frac{i\rho\,\ve^{2}}{\pi(\xi_{3}(\ve)-\xi_{3}^{*}(\ve))}K(x,x')e^{-i\,b_{\ve}\,t}e^{-\,\gamma_{\ve}\,t}
\] 
with
\[
K(x,x')=\frac{e^{-(4\pi)^{2}|\al|(|x|+|x'|)}}{|x|\,|x'|},\quad\rho=\frac{|\al|\sqrt{2\beta-(4\pi|\al|)^{2}}}{\xi_{1}(\ve)\,\xi_{2}(\ve)}
\]
and the following estimate holds true
\[ 
%\label{orabasta}
\|B(t,\ve)\|_{HS}\leq C \ve^{2}\,.
\]
\end{theorem}
\begin{proof}
The integrand function appearing in (\ref{subst1}) can be analytically continued to $Q$, and by the residues theorem we get
\begin{equation*}
\mathscr P_{I,++}(t):=P_1(t)-P_2(t)
\end{equation*}
where
\begin{equation*}
P_1(t;x,x')=-\frac{i\ve^2\,e^{-i\,b_{\ve}\,t}e^{-\,\gamma_{\ve}\,t}}{\pi(\xi_{3}(\ve)-\xi_{3}^{*}(\ve))}K(x,x')
\frac{ f(\xi_{3}(\ve))
\,e^{-(4\pi)^{2}\xi_{3}(\ve)(|x|+|x'|)}}
{(\xi_{3}(\ve)-\xi_{1}(\ve))(\xi_{3}(\ve)-\xi_{2}(\ve))}
\end{equation*}
and
\begin{equation*}
P_2(t;x,x')=
-\frac{\ve^2}{2\pi^2}K(x,x')
\int_{\partial Q\backslash I^{'}}\frac{ f(z)\,e^{-i\,\left(\beta-(4\pi)^{2}(z+|\al|)^{2}\right)\,t}
\,e^{-(4\pi)^{2}z(|x|+|x'|)}}
{(z-\xi_{1}(\ve))(z-\xi_{2}(\ve))(z-\xi_{3}(\ve))(z-\xi_{3}^{*}(\ve))}
dz\,.
\end{equation*}
Where we used $\xi_3=\sqrt{\beta-b_\ve+i\ga_\ve}-|\al|$. Obviously the operators $P_1(t)$ and $P_2(t)$ depend on $\ve$ even if it is not indicated explicitly. 

We analyze first the term $P_2(t,x,x')$. For all $x,\,x'\in\RE^3$ and for $\ve$ small enough
\begin{equation*}
|P_2(t,x,x')|\leqslant \ve^2 C \,|K(x,x')| e^{-(4\pi)^{2}\inf(I^{'})(|x|+|x'|)}\,.
\end{equation*}
Where we used the fact that for all $z\in \partial Q\backslash I^{'}$, $\Im\big((z+|\al|)^2\big)>0$, and that for $\ve$ small enough 
\[
|f(z)|\leqslant C\qquad\textrm{and}\qquad
|(z-\xi_{1}(\ve))(z-\xi_{2}(\ve))(z-\xi_{3}(\ve))(z-\xi_{3}^{*}(\ve))|^{-1}\leqslant C\,.
\]
Since $|\al|+\inf(I^{'})>0$ we get the estimate $\|P_2(t)\|_{HS}\leqslant C\ve^2$\,.

From the definition of $L(t,\ve)$ and $P_1(t)$ we have 
\[
\begin{aligned}
 |L(t,\ve;x,x')-P_1(t;x,x')|=&
-\frac{i\ve^2\,e^{-i\,b_{\ve}\,t}e^{-\,\gamma_{\ve}\,t}}{\pi(\xi_{3}(\ve)-\xi_{3}^{*}(\ve))}K(x,x')\\
&\bigg[
\frac{|\al|\sqrt{2\beta-(4\pi|\al|)^{2}}}{\xi_{1}(\ve)\,\xi_{2}(\ve)}
-\frac{ f(\xi_{3}(\ve))
\,e^{-(4\pi)^{2}\xi_{3}(\ve)(|x|+|x'|)}}
{(\xi_{3}(\ve)-\xi_{1}(\ve))(\xi_{3}(\ve)-\xi_{2}(\ve))}
\bigg]\,.
\end{aligned}
\]
The estimate $\|L(t,\ve)-P_1(t)\|_{HS}\leqslant C\ve^2$ comes directly from $|\xi_3(\ve)|\leqslant C\ve^2$. 

The proof of the theorem is then obtained by setting $B(t,\ve)=P_1(t)-L(t,\ve)-P_2(t)$.
\end{proof}

As pointed out by many authors (see, e.g., \cite{EXN} and references therein) the estimate given in theorem \ref{teorema4} does not make explicit the short time behavior of the solution. Notice in fact that the error term $B(t,\ve)$ might be in principle much larger than $|L(0,\ve)-L(t,\ve)|$ for 
$0 \leq t \leq 1$. The following representation of 
$\mathscr P_{I,++}(t)$ is more suitable to examine details of the short time evolution of the projector.
\begin{proposition}
\label{teorema5}
Let  $\mathscr P_{I,++}(t)$ be defined like in (\ref{basta}) and $0\leqslant t \ll 1/\ve$, then
\be
\label{uffa}
\mathscr P_{I,++}(t;x,x')= a(x,x') -b(x,x')(1-e^{- i \; E_\ve^{res}\;t}) +c(x,x') \ve^{2} t +\OO( \ve^{2} t^2) d(x,x')
\ee
where 
\[
a(x,x') = -\frac{\ve^2}{2\pi^2}
K(x,x')
\int_{I^{'}}
\frac{ f(\xi)\,e^{-(4\pi)^{2}\xi(|x|+|x'|)}}
{\xi^{4}+2|\al|\xi^{3}-\frac{2\beta}{(4\pi)^{2}}\xi^{2}-2\ve^{2}|\al|\xi-\ve^{4}}
d\xi
\]
\[
b(x,x') = -\frac{\ve^2}{2\pi^2\;(\xi_{3}(\ve)-\xi_{3}^{*}(\ve))}K(x,x')
\frac{f(\xi_3(\ve))\,e^{-(4\pi)^{2}\xi_{3}(\ve)(|x|+|x'|)}}
{(\xi_{3}(\ve)-\xi_{1}(\ve))((\xi_{3}(\ve)-\xi_{2}(\ve))}
\]
\[
c(x,x') = \frac{1}{2\pi^2}K(x,x')
\int_{\partial Q\backslash I^{'}}
\frac{ f(z)\,i\; (\bt-(4\pi)^{2}(z+|\al|)^{2}) \; e^{-(4\pi)^{2}z (|x|+|x'|)}}
{z^{4}+2|\al|z^{3}-\frac{2\beta}{(4\pi)^{2}}z^{2}-2\ve^{2}|\al|z-\ve^{4}}
dz\,.
\]
Moreover
\begin{equation}
\label{bach}
\|a\|_{HS}\leqslant C\;;\qquad\|b\|_{HS}\leqslant C\;;\qquad\|c\|_{HS}\leqslant C
\;;\qquad\|d\|_{HS}\leqslant C\,.
\end{equation}
\end{proposition}
\begin{proof}
The proof is obtained by writing $\mathscr P_{I,++}(t)$ as $\mathscr P_{I,++}(0)+(\mathscr P_{I,++}(t)-\mathscr P_{I,++}(0))$. Then $a=\mathscr P_{I,++}(0)$ and  
\[
\mathscr P_{I,++}(t)-\mathscr P_{I,++}(0)=-
\frac{\ve^2}{2\pi^2}K(x,x')\,
\int_{I^{'}}\frac{ f(\xi)\,\big(e^{-i\,\left(\beta-(4\pi)^{2}(\xi+|\al|)^{2}\right)\,t}-1\big)
\,e^{-(4\pi)^{2}\xi(|x|+|x'|)}}
{\xi^{4}+2|\al|\xi^{3}-\frac{2\beta}{(4\pi)^{2}}\xi^{2}-2\ve^{2}|\al|\xi-\ve^{4}}
d\xi\,.
\]
Following what was done in the proof of theorem \ref{teorema4} we can use the analytic continuation theory and the residues theorem. Then the term $b$ in \eqref{uffa} corresponds to the contribution from the resonant pole whereas the third term is the derivative in $t=0$ of the integral on the semicircular part of the boundary of $Q$. 

The estimates \eqref{bach} can be easily obtained by an argument similar to the one used in theorem \eqref{teorema4}, for this reason we omit the details. 
\end{proof}

Notice that none of the integral kernels in the representation of $\mathscr P_{I,++}(t)$ given in the proposition correspond to a projection operator.

\vspace{0.3cm}

{\bf Acknowledgments} The authors are grateful to Sergio Albeverio and 
Pavel Exner for the useful comments and for pointing out most of the bibliography  related to the theory of resonances. This work was supported by the Doppler Institute grant (LC06002).\\

\end{document}